\documentclass[10pt]{article}

\usepackage{hyperref}
\hypersetup{colorlinks=true,linkcolor=red,citecolor=green,filecolor=magneta,urlcolor=cyan}
\usepackage{amsthm}
\usepackage{amsmath}
\usepackage{amsfonts}

\usepackage{pgf}

\usepackage{amssymb}

\usepackage{amsthm}
\usepackage{amsmath}

\newtheorem{theorem}{Theorem}[section]

\newtheorem{lemma}[theorem]{Lemma}

\newtheorem{property}[theorem]{Property}
\theoremstyle{definition}
\newtheorem{definition}[theorem]{Definition}
\newtheorem{notation}[theorem]{Notation}
\newtheorem{remark}[theorem]{Remark}
\newtheorem{example}[theorem]{Example}

\begin{document}



\title{Suppressing chaos in discontinuous systems \\of fractional order by active control}


\author{
Marius-F Danca \\
\small Dept. of Mathematics and Computer Science, \\ 
\small Avram Iancu University, 400380 Cluj-Napoca, Romania \\
\small Romanian Institute of Science and Technology, 400487 Cluj-Napoca, Romania \\
\small \texttt{danca@rist.ro}
\and
Roberto Garrappa \\
\small Dept. of Mathematics, University of Bari, 70125 Bari, Italy\\ 
\small \texttt{roberto.garrappa@uniba.it}
}

\date{May 2, 2014}

\maketitle

\begin{abstract}
In this paper, a chaos control algorithm for a class of piece-wise continuous chaotic systems of fractional order, in the Caputo sense, is proposed. With the aid of Filippov's convex regularization and via differential inclusions, the underlying discontinuous initial value problem is first recast in terms of a set-valued problem and hence it is continuously approximated by using Cellina's Theorem for differential inclusions. For chaos control, an active control technique is implemented so that the unstable equilibria become stable. As example, Shimizu--Morioka's system is considered. Numerical simulations are obtained by means of the Adams-Bashforth-Moulton method for differential equations of fractional-order.

\end{abstract}



\section{Introduction}
\label{sec1}

Several real--life systems show non--smooth physical properties (for instance dry friction, forced vibration brake processes with locking phase, stick, and slip phenomena) which can be suitably modeled by introducing some kind of discontinuity. Moreover, anomalous processes (for instance in non--standard materials) exhibit \emph{memory} and \emph{ereditary} properties and derivatives of fractional order are an effective tool to keep into account these phenomena.

Thus, fractional discontinuous systems provide a logical and attractive link between systems of fractional order and discontinuous systems.

Chaos control in continuous fractional-order systems, have been realized for many systems such as: Lorenz system, Chua system, Rossler system, Chen system, Liu system, Rabinovich--Fabrikant system, Coullet system, dynamos system, Duffing system, Arneodo system, Newton-Leipnic system and so on (few of the numerous related papers are \cite{Srivastava2013,Abd-Elouahab2010,Razminia2011,Richter2002,LiChen2013}).

Anyway, because of the lack of numerical methods specifically devised for fractional differential equations (FDEs) with discontinuous right-hand side, discontinuous systems of fractional-order have not been rigorously studied. 

In \cite{Garrappa2014} it was investigated the behavior of some classical methods for discontinuous FDEs which a set--valued regularization into the Filippov's framework is applied \cite{Filippov1988,DieciLopez2009}; in particular, the chattering--free behavior of the generalization of the implicit Euler scheme was showed.

Solving problems resulting from the Filippov set--valued regularization is however not a simple task and requires, in most cases, to recast the original problem in a linear complementarity problem whose numerical solution is often very demanding.

A possible way to remove these obstacles, is to approximate continuously the underlying initial value problem by modeling the discontinuous system according to the algorithm presented in \cite{dan0}. In this way the chaos control problem becomes a standard chaos control of continuous systems of fractional-order.

In this paper we focus on discontinuous problems in which the right--hand side is a piece-wise continuous (PWC) function $f:\mathbb{R}^n\rightarrow \mathbb{R}^n$ having the following form
\begin{equation}\label{f}
f(x(t))=g(x(t))+Kx(t)+A(x(t))s(x(t)),
\end{equation}

\noindent where $g:\mathbb{R}^n\rightarrow \mathbb{R}^n$ is a nonlinear and (at least) continuous function, $s:\mathbb{R}^n\rightarrow \mathbb{R}^n$ a piece-wise continuous function $s(x)=(s_1(x_1),s_2(x_2),...,s_n(x_n))^T$ with $s_i:\mathbb{R}\rightarrow \mathbb{R}$, $i=1,2,...,n$ real piece-wise constant functions, $A\in \mathbb{R}^{n\times n}$ a square matrix of real continuous functions and $K$ a square real matrix, representing the linear part of $f$.

\vspace{3mm}



Discontinuous systems of fractional-order are modeled in this paper by the following initial value problem (IVP)
\begin{equation}\label{IVP0}
	D_*^q x(t)=f(x(t)), \quad x(0)=x_0, \quad t\in I=[0,\infty),
\end{equation}
where $f$ is the PWC function defined by (\ref{f}) and, for $q=(q_1,q_2,\dots,q_n)$, $D_*^q x(t) = (D_*^{q_{1}} x_1(t), D_*^{q_2} x_2(t), \dots, D_*^{q_n} x_n(t))$ denotes the vector of the differential operator of fractional order $q_i$ applied to each component of $x(t)$.

In the past several alternative definitions have been proposed to provide valuable generalizations of differential operators to non integer order. Although the approach named as Riemann--Liouville is the most important both for theoretical and historical reasons, for practical applications the definition due to Caputo \cite{Caputo1969} is the most appropriate and useful. Indeed, the Caputo's fractional derivative has the considerable advantage of allowing to couple differential equations with classical initial conditions of Cauchy type as in (\ref{IVP0}) which not only have a clearly interpretable physical meaning \cite{HeymansPodlubny2006} but can also be measured to proper initializing the simulation.

Since chaotic fractional-order systems are usually modeled with subunit fractional orders $0<q_i\le 1$, $i=1,2,\dots,n$, the Caputo's differential operator of order $q_{i}$, with respect to the starting point 0, is defined as \cite{Caputo1969,OldhamSpanier1974,Podlubny1999}
\begin{equation*}
D_*^{q_i}x_{i}(t)=\frac{1}{\Gamma(1-q_{i})}\int_0^t (t-\tau)^{-q_{i}}\frac {d}{dt}x_{i}(\tau)d\tau,
\end{equation*}
where $\Gamma(z)$ is the Euler's Gamma function.

Regarding the matrix $A$, the following assumption is considered

\vspace{3mm}
\noindent (\textbf{H1}) $A(x)s(x)$ is discontinuous in at least one of his components.
\vspace{3mm}

The discontinuity impediment can be avoided, by using the Filippov's technique \cite{Filippov1988} to convert a single valued (discontinuous) problem into a set-valued one. Then, via Cellina's Theorem for differential inclusions \cite{AubinCellina1984,AubinFrankowska1990}, set-valued functions can be continuously approximated in small neighborhoods.

The continuous approximation algorithm proposed in this paper regards the discontinuous functions $s_i$ being valid for a large class of functions such as the \emph{Heaviside} function, the \emph{rectangular} function (as difference of two Heaviside functions), or the \emph{signum}, one of the most encountered PWC functions in practical applications.

Standard techniques can hence be applied to the continuous approximation to device an active control in order to suppress the appearance of chaos.

The paper is organized as follows: Section \ref{switch} deals with the approximation of the PWC function (\ref{f}) and shows how the IVP (\ref{IVP0}) can be transformed into a continuous problem. Section \ref{S:notions} concerns the investigation of stability issues of the approximated continuous problem. In Section \ref{S:Stabilizing} the chaos control obtained by stabilizing unstable equilibria of PWC systems of fractional-order is investigated. The application of these techniques to the fractional-order variant of Shimizu-–Morioka's system (\ref{simiz}) is hence analyzed in Section \ref{S:Applications}.

\section{Continuous approximation of PWC systems of fractional-order}\label{switch}

\noindent\textbf{Notation}
Let $\mathcal{M}$ be the discontinuity set of $f$, generated by the discontinuity points of the components $s_i$.
\vspace{3mm}

\begin{example}
For the linear PWC function $f:\mathbb{R}\rightarrow \mathbb{R}$
\begin{equation}\label{exemplu}
f(x)=2-3sgn(x),
\end{equation}
the discontinuity set is $\mathcal{M}=\{0\}$ and determines on $\mathbb{R}$ the continuity sub-domains $\mathcal{D}_1=(-\infty, 0]$ and $\mathcal{D}_2=[0,\infty)$ (see Figure \ref{fig1}).
\end{example}

The example of PWC systems analyzed in this paper, is the fractional variant of the chaotic Shimizu--Morioka's three-dimensional system \cite{ShimizuMorioka1980,YuTangLuChen2009}
\begin{equation}\label{simiz}
\begin{array}{l}
D_{\ast }^{q_{1}}x_{1}=x_{2}, \\
D_{\ast }^{q_2} x_{2}=(1-x_3)sgn(x_{1})-a x_2 , \\
D_{\ast }^{q_3} x_{3}=x_1^2-b x_3,%
\end{array}%
\end{equation}

\noindent with $a=0.75$ and $b=0.45$. Here $g(x)=(0,0,x_1^2)$ and
\[
K=\left(
\begin{array}{ccc}
0 & 1 & 0 \\
0 & -a & 0 \\
0 & 0 & -b%
\end{array}%
\right),~~~
A(x)=\left(
\begin{array}{ccc}
0 & 0 & 0 \\
1-x_{3} & 0 & 0 \\
0 & 0 & 0%
\end{array}%
\right).
\]

The chaotic behavior is revealed in the bifurcation diagrams where the extrema of the state variables are plotted for $a=0.75$ and with respect to $b \in [0,1]$; the commensurate case $q_1=q_2=q_3=0.95$ is presented in the left column of Figure \ref{bd} whilst the incommensurate case $q_1=1$, $q_2=q_3=0.9$ is plotted in the right column of Figure \ref{bd}. As it is well--known, in fractional-order systems chaos exists even for order $q_1+q_2+q_3<3$.

\begin{remark}\label{chendis}
The discontinuity in this example is due only to the component $sgn(x_1)$, the other component, $x_3sgn(x_1)$, being non-smooth, but continuous \emph{(}compare with Assumption \emph{\textbf{H1}}\emph{ and \cite{Danca2007})}.
\end{remark}

The class of PWC functions $f$ defined in (\ref{f}), can be approximated as closely as desired, with continuous functions. To this purpose, $f$ is first transformed into a set-valued convex function $F:\mathbb{R}^n\rightrightarrows \mathbb{R}^n$, assuming values into the set of all subsets of $\mathbb{R}^n$, via the so called \emph{Filippov regularization} \cite{Filippov1988}
\begin{equation}\label{fill}
F(x)=\bigcap_{\varepsilon >0}\bigcap_{\mu(\mathcal{M})=0} \overline{conv}(f({z\in \mathbb{R}^n: |z-x|\leq\varepsilon}\backslash \mathcal{M})).
\end{equation}

$F(x)$ is the convex hull of $f(x)$, $\varepsilon$ being the radius of the ball centered in $x$ (see the sketch in Figure \ref{fig2}, where $\varepsilon$ has been taken quite large for a clear understanding).

If $s_i$ are $sgn$ functions, the underlying set-valued form, denoted by $Sgn:\mathbb{R}\rightrightarrows \mathbb{R}$, is defined as follows (see Figure \ref{fig3})
\begin{equation}
Sgn(x)=\left\{
\begin{array}{cc}
\{-1\}, & x<0, \\
\lbrack -1,1], & x=0, \\
\{+1\}, & x>0.%
\end{array}%
\right.
\end{equation}

The Filippov regularization applied to $f$ leads to the following set-valued function
\begin{equation}
\label{IVP1}
F(x)=g(x)+Kx+A(x)S(x),
\end{equation}

\noindent with
\begin{equation}\label{s}
S(x)=(S_1(x_1),S_2(x_2),...,S_n(x_n))^T
\end{equation}

\noindent and where $S_i:\mathbb{R}\rightarrow \mathbb{R}$ are the set-valued variants of $s_i$, $i=1,2,...,n$ ($Sgn(x_i)$ for the usual case of $sgn(x_i)$).

The following notions and results are presented in $\mathbb{R}$, but they are also valid in the general case $\mathbb{R}^n$, $n>1$.

\begin{definition}
A single-valued function $h:\mathbb{R}\rightarrow \mathbb{R}$ is called an \emph{approximation} (\emph{selection}) of the set-valued function $F:\mathbb{R}\rightrightarrows \mathbb{R}$ if
\begin{equation*}
\forall	x\in \mathbb{R}, \quad h(x)\in F(x).
\end{equation*}
\end{definition}

The approximations can be done locally or globally \cite{dan0}. Usually, a set-valued function admits (infinitely) many local or global approximations; we refer to Figure \ref{fig2} for the case of the function (\ref{exemplu}).

In this paper we will use global approximations, which are easier to implement numerically.

\begin{lemma}\label{tprinc}\emph{\cite{dan0}}
For every $\varepsilon>0$, the set-valued functions $S_i$, $i=1,2,...,n$ admit continuous global approximations in the $\varepsilon$-neighborhood of $S_i$.
\end{lemma}
\begin{proof} $S_i$, for $i=1,2,...,n$, verify Cellina's Theorem \cite{AubinCellina1984,AubinFrankowska1990} which ensures the existence of continuous approximations of $S_i$ on $\mathbb{R}$ (the conditions required by Cellina's Theorem are verified via the Remark in \cite{Filippov1988} p. 43 and the Example in \cite{AubinFrankowska1990} p. 39).
\end{proof}

\begin{notation}
Let denote by $\widetilde{s}_i:\mathbb{R}\rightarrow \mathbb{R}$ the global approximations of $S_i$.
\end{notation}

Let consider, for the sake of simplicity, that for each component $\widetilde{s}_i(x_i)$, $i=1,2,...,n$, $\varepsilon_i$ have the same value.

Since, most of practical examples of PWC systems are modeled via $sgn$ function, we will use for its approximation one of the so called \emph{sigmoid functions}
\begin{equation}\label{h_simplu}
\widetilde{sgn}(x)=\frac{2}{1+e^{-\frac{x}{\delta}}}-1\approx Sgn(x)
\end{equation}
because this class of functions provide the required flexibility being the abruptness of the discontinuity easily adaptable.\footnote{The class of sigmoid functions includes for example the arctangent such as $\frac{2}{\pi}arctan\frac{x}{\varepsilon}$, the hyperbolic tangent, the error function, the logistic function, algebraic functions like $\frac{x}{\sqrt{\epsilon+x^2}}$, and so on.} $\delta$ is a positive parameter which controls the slope in the neighborhood of the discontinuity $x=0$ and determines the $\varepsilon$-neighborhood size (see Figure \ref{fig4}, where $\widetilde{sgn}$ curves are represented as functions of $\delta$).

The technical Lemma \ref{tprinc} allows us to introduce the following result

\begin{theorem}\label{th1}\emph{\cite{dan0}}
Let $f$ defined by \emph{(\ref{f})}. There exist continuous global approximations $\tilde{f}:\mathbb{R}^n\rightarrow \mathbb{R}^n$ of $f$ defined as
\begin{equation}\label{glo}
\tilde{f}(x)=g(x)+Kx+A(x)\widetilde{s}(x)\approx f(x).
\end{equation}
\end{theorem}

For example, when $f$ is the function defined in (\ref{exemplu}), it be approximated as 
\begin{equation}
\tilde{f}(x)=2-3\widetilde{sgn}(x)=2-3\left(\frac{2}{1+e^{-\frac{x}{\delta}}}-1\right).
\end{equation}

Once Theorem \ref{th1} has been established, we can enunciate the main result of this section.

\begin{theorem}\label{th2}\emph{\cite{dan0}}
\noindent The IVP \emph{(\ref{IVP0})} can be continuously approximated by the following continuous IVP
\begin{equation}\label{doi}
D_*^q x(x)=\tilde{f}(x),~~~x(0)=x_0,
\end{equation}
\end{theorem}

Thanks to Theorem \ref{th2}, we are able to transform a discontinuous control problem into a continuous one of fractional-order, where known control methods can apply.

\section{Utilized notions and results}\label{S:notions}

In this section some related properties related to $\tilde{f}$, are presented and discussed.

Let us consider the IVP (\ref{doi}) with $f$ defined by (\ref{f}), and denote with $X^*$ and $\tilde{X}^*$ the equilibrium points of $f$ and $\tilde{f}$ respectively, and with $J$ and $\tilde{J}$ their related Jacobians.

The computation of the Jacobi matrix required in the stability analysis, imposes the following assumption
\vspace{3mm}

\noindent (\textbf{H2})
The function $g$ in (\ref{f}) is supposed to be differentiable on $\mathbb{R}^n$.

\begin{property}\label{prop1}\emph{\cite{danx}}
$\tilde{X}^*\approx X^*$.
\end{property}

\begin{remark}\label{cater}
As we will better explain later on, for the numerical simulations we will use a method of  order $O(h^p)$, $p = min(2,1+q_{min})$, with a step--size $h=0.005$ and hence providing an error proportional to $10^{-5}$. Therefore, numerically, Property \emph{\ref{prop1}} reads as follows: choosing for example $\delta=1/100000$, for $x\notin \mathcal{V}_\delta=(-1.589\times 10^{-4},1.589\times 10^{-4})$, the difference between $\widetilde{sgn}$ and the branch $\pm1$ of the function $sgn$ is of order of $10^{-7}$, which implies $\tilde{X}^*\approx X^*$ (see \emph{Fig. \ref{fig5}}, where for clarity, $\delta=1/2$).
\end{remark}

In view of the above remark, For numerical reasons the following hypothesis is assumed throughout the paper.

\vspace{3mm}
\noindent (\textbf{H3})\label{del} In the numerical simulations $\delta=1/100000$.
\vspace{3mm}

Regarding the determination of $\tilde{J}$, the following property holds.

\begin{property}\emph{\cite{danx}}\label{prop2}
Assume \textup{(\textbf{H2})}. Then $\tilde{J}|_{\widetilde{X}^*}\approx J|_{X^*}$.
\end{property}

\noindent The proof, presented in \cite{danx}, is based on the approximation $\frac{d}{dx}sgn\approx \frac{d}{dx}\widetilde{sgn}$.

\vspace{3mm}

For the sake of brevity, hereafter, aided by Properties \ref{prop1} and \ref{prop2}, equilibria and Jacobians of the approximated system (\ref{glo}) will be denoted simply by $X^*$ and $J$ respectively.

As known, a fractional-order (continuous) system is asymptotically stable at some of his equilibria $X^*$ if $X^*$ is asymptotically stable.

The following theorem, which encompasses the asymptotically stability results for systems of commensurate and incommensurate order \cite{Matignon1998,DengLiLu2007,TavazoeiHaeriJafari2008}, states the necessary and sufficient conditions for stability for our class of fractional PWC systems

\begin{theorem}\label{stabt}
Let $X^*$ be an equilibrium point of the PWC system of fractional-order \emph{(\ref{IVP0})}. $X^*$ is asymptotically stable if and only if:
\itemize
\item[(i)] for the commensurate case $q_1=q_2=...=q_n=q$, all eigenvalues $\lambda$ of the Jacobian $J_{X^*}$ evaluated at $\tilde{X}^*$, verify the condition
\begin{equation}\label{stab1}
|arg(\lambda)|>q \pi/2;
\end{equation}
\item[(ii)] for the incommensurate case $q_i=n_i/m_i<1$, $n_i$, $m_i\in \mathbb{N}$, $m_i\neq 0$, for $i=1,2,...,n$ with $n_i$ and $m_i$ coprime integers, $(n_i,m_i)=1$, all the roots $\lambda$ of the characteristic equation
\begin{equation}\label{car}
P(\lambda):=det(diag [\lambda^{mq_1},\lambda^{mq_2},...,\lambda^{mq_n}]-J_{X^*})=0,
\end{equation}

\noindent with $m$ the least common multiple of the denominators $m_i$, verify the condition
\begin{equation}
\label{stab2}
|arg(\lambda)|> \pi/2m.
\end{equation}

\enditemize

\end{theorem}

\begin{proof} Using Theorem \ref{th2}, Assumption \textbf{H2} and Properties \ref{prop1} and
\ref{prop2}, the system transforms into a continuous systems of fractional-order with equilibria $X^*$ and Jacobian $J_{X^*}$. Therefore, the proof can be done as for
the underlying theorems for continuous commensurate and incommensurate systems (see e.g. \cite{Matignon1998,DengLiLu2007}).
\end{proof}

The stability for the commensurate case can be considered as a particular case (corollary) of the incommensurate case.

If we denote with $\Lambda$ the set of the roots of the polynomial $P$ defined in (\ref{car}), and with $\alpha_{min}=min\{|arg(\Lambda)\}$, the sufficient and necessary asymptotically stability conditions (\ref{stab1}) and (\ref{stab2}) can be written in the following compact working form
\begin{equation}
\label{stab3}
\alpha_{min}> \gamma\pi/2,
\end{equation}

\noindent where $\gamma=q$ for the commensurate case, and $\gamma=1/m$ in the case of incommensurate case.

The domain

\begin{equation}\label{staba}
\Omega=\{\lambda\in \mathcal{C} | |arg(\lambda)|\geq \gamma\pi/2\},
\end{equation}

\noindent is called the stability domain.

\begin{remark}\label{unstab}
\item[i)]From Theorem \ref{stabt} we can deduce the instability condition of the equilibrium point $X^*$
\begin{equation}
\label{inec_u}
\alpha_{min}\leq \gamma\pi/2,
\end{equation}
\noindent which means that at least one eigenvalue is outside the stability domain $\Omega$.
\item[ii)] For the commensurate case ($\gamma=q$) of the instability condition (\ref{inec_u}), one can obtain the minimal value $q_{min}$ for which the system can generate chaos
\begin{equation}\label{min}
q>q_{min}=\frac{2}{\pi}\alpha_{min}.
\end{equation}

\item[iii)]Relation (\ref{inec_u}) gives for the commensurate case a necessary condition for chaotic behavior
\begin{equation}
\label{nece}
\alpha_{min}\leq q\pi/2.
\end{equation}
 \noindent This means that if some equilibrium point $X^*$ is unstable, then chaos is possible, but  (\ref{nece}) does not imply the chaos presence for all initial conditions, because some other equilibrium point, $X^{**}$, might be stable and in this case initial conditions can be found in the attraction basin of $X^{**}$ for which the underlying trajectory will not reach the chaotic attractor (as known, the coexistence of chaotic attractors and stable equilibrium points is possible). However, reversely, if the underlying system poses a chaotic attractor for all initial conditions, then condition (\ref{nece}) is certainly verified.

\end{remark}

\section{Stabilizing unstable equilibria}\label{S:Stabilizing}

Let now consider the chaotic fractional-order PWC system (\ref{IVP0}), with his continuous approximation (\ref{glo}), having at least one unstable equilibrium $X^*$. In order to control the chaos (stabilization of $X^*$), we must choose an active state feedback controller $u$, such that the controlled system
\begin{equation}\label{csys}
D^q_*x(t)=\tilde{f}(x(t),u(t))=g(x(t))+Kx(t)+A(x(t))\tilde{s}(x(t))+u(t),
\end{equation}
can be driven to reach asymptotically the control target $X^*$ (i.e. a classical synchronization problem). Since for the uncontrolled system, $X^*(x_1^*,x_2^*,\ldots,x_n^*)$ is a solution, it verifies the equation (\ref{csys}) for $u(t)=0$, $t\in I$
\begin{equation}\label{csys_fara}
D^q_*x^*(t)=\tilde{f}(x^*(t),0)=g(x^*(t))+Kx^*(t)+A(x^*(t))\tilde{s}(x^*(t)),
\end{equation}
and the controller $u$ has to be chosen such that $\lim_{t \to \infty} e(t)=\lim_{t \to \infty} ||x(t)-x^*(t)||=0$. The error state $e$ can be obtained by subtracting (\ref{csys_fara}) from the controlled system (\ref{csys}). Beside the state $e$, the obtained system contains nonlinear parts. To countervail his effect, next we have to design $u$ in the following form
\begin{equation}\label{ctrl}
u(t)=-g(x(t))-A(x(t))\widetilde{s}(x(t))+g(x^*(t))+A(x^*(t))\widetilde{s}(x^*(t))+v(t),
\end{equation}
such as the obtained system become a linear (error) system of fractional-order.

$v$, which can be considered as an external input, is a linear controller, and concerns the stabilization of the obtained linear system, having the form $v(t)=Me(t)$, with $e(t)=(e_1(t),e_2(t),...,e_n(t))^T$, $e_i(t)=x_i(t)-x_i^*(t)$ and $M\in R^{n\times n}$ some real square gain matrix.

Finally, the obtained error system, which describes the error dynamics, has the following form
\begin{equation}\label{liniar}
D^q_*e(t)=Ee(t).
\end{equation}
and the error matrix $E$ can be evaluated by the following relation
\[
E=K+M.
\]

As known from stability theory of fractional-order systems, the linear autonomous system of fractional-order (\ref{liniar}) is asymptotically stable if his zero (equilibrium point) is asymptotically stable. Therefore, the stabilization of $X^*$ certified via Theorem \ref{stabt}, applies to (\ref{liniar}).

\begin{remark}
\item[i)]If $X^*$ is an equilibrium, then $x^*_i(t)$ are constant: $x^*_i(t)=x_i^*$ for all $t\in I$.
\item[ii)]If $X^*$ is a periodic solution, then the chaos control becomes a classical synchronization problem between systems (\ref{csys}) and (\ref{csys_fara}).
\end{remark}

\section{Applications}\label{S:Applications}

In this section the chaotic behavior of Shimizu--Morioka's system (\ref{simiz}), generated by the instability of system's equilibria is controlled. The commensurate case $q=(0.95,0.95,0.95)$ and incommensurate case $q=(1,0.9,0.9)$, are considered.

To draw bifurcation diagrams, phase plots and time series we make use of the Matlab code {\tt fde12.m} \cite{Garrappa_FDE12} which has been suitably modified to deal with the incommensurate case too. This code implements the predictor--corrector method described in \cite{DiethelmFordFreed2002,FordSimpson2001} and based on product--integration rules of Adams--Bashforth--Moulton type \cite{Young1954}. A step--size $h=0.005$ has been used to generate the time--series. Since the order of the method is ${\cal O}(h^p)$, $p = min(2,1+\min\{q_{i}\})$, an error proportional to $10^{-5}$ is expected. A smaller step--size $h=2^{-9}\approx 0.002$ has been instead used for the bifurcation diagrams in order to improve the clearness of the plots. To approximate the roots of the characteristic equations the Matlab built--in function {\tt solve} has been employed with an accuracy up to three decimals.

For all considered cases, the control is activated at $t=50$.

\noindent After approximation, the system becomes (see Remark \ref{chendis})
\begin{equation}\label{sim_ne}
\begin{array}{l}
D_{\ast }^{q_{1}}x_{1}=x_{2}, \\
D_{\ast }^{q_2} x_{2}=\widetilde{sgn}(x_1)-x_3sgn(x_{1})-a x_2 , \\
D_{\ast }^{q_3} x_{3}=x_1^2-b x_3.
\end{array}%
\end{equation}

\noindent Via Property \ref{prop1}, and (\ref{prop2}) the system has three equilibria: $X^*_{1,2}=(\pm \sqrt{b},0,1)$ and $X_3^*=(0,0,0)$, and its Jacobian is
\begin{equation}\label{jacob}
J_{X^*}=\left(
\begin{array}{ccc}
0 & 1 & 0 \\
0 & -a & -sgn(x_1^*) \\
2x_1^* & 0 & -b
\end{array}%
\right),
\end{equation}

\noindent with $x_1^*=0,\pm\sqrt{b}$. The controlled system is
\begin{equation}\label{sim_master}
\begin{array}{l}
D_{\ast }^{q_{1}}x_{1}=x_{2}+u_1, \\
D_{\ast }^{q_2} x_{2}=\widetilde{sgn}(x_1)-x_3sgn(x_{1})-a x_2+u_2 , \\
D_{\ast }^{q_3} x_{3}=x_1^2-b x_3+u_3,%
\end{array}%
\end{equation}

\noindent with the control (\ref{ctrl}) in the following form
\begin{equation}\label{ctrl_simi}
\left(
\begin{array}{c}
u_{1} \\
u_{2} \\
u_{3}%
\end{array}%
\right)=\left(
\begin{array}{c}
0\\
-\widetilde{sgn}(x_{1})+x_{3}{sgn}(x_{1})+\widetilde{sgn}(x_{1}^*)-x_{3}^*sgn(x_{1}^*) \\
-x_{1}^{2}+(x_1^*)^2
\end{array}%
\right) +\left(
\begin{array}{c}
v_{1} \\
v_{2} \\
v_{3}%
\end{array}%
\right).
\end{equation}

\noindent where $(v_1,v_2,v_3)^T$ will be chosen depending on $X_{1,2,3}^*$ coordinates.

\vspace{3mm}

\noindent \emph{\underline{Commensurate case}} $q=(0.95,0.95,0.95)$.

\noindent\emph{Equilibrium} $X_{1}^*(\sqrt{0.45},0,1)$. The eigenvalues spectrum is $\Lambda=\{0.172\pm0.916i,-1.544\}$ with arguments $ \{\pm1.385,\pi\}$. Because $\alpha_{min}=1.385< 1.492=q\pi/2=0.95\pi/2$, $X_1^*$ is unstable (relation (\ref{inec_u})). From the condition (\ref{min}), the instability persists for $q>0.882$. For the chosen value $q=(0.95,0.95,0.95)$, numerical simulations reveals that the system behaves chaotically, being in accord with the necessary condition (\ref{nece}) (Fig. \ref{fig6} a).

\noindent For $X_1^*$, the controller (\ref{ctrl_simi}) becomes
\begin{equation}\label{ccc1}
\left(
\begin{array}{c}
u_{1} \\
u_{2} \\
u_{3}%
\end{array}%
\right)=\left(
\begin{array}{c}
0 \\
-\widetilde{sgn}(x_{1})+x_{3}{sgn}(x_{1}) \\
-x_{1}^{2}+0.45
\end{array}%
\right) +\left(
\begin{array}{c}
v_{1} \\
v_{2} \\
v_{3}%
\end{array}%
\right).
\end{equation}

\noindent In this case $e_1=x_1-x_1^*=x_1-\sqrt{0.45}, e_2=x_2-x_2^*=x_2$, $e_3=x_3-x_3^*=x_3-1$ and, if for $M$ we choose the following form
\[
M=\left(
\begin{array}{ccc}
-1 & 0 & 0 \\
0 & 0 & 0 \\
0 & 0 & 0%
\end{array}\right),
\]

\noindent $v$ becomes
\begin{equation}\label {v}
\left(
\begin{array}{c}
v_1 \\
v_2 \\
v_3%
\end{array}%
\right)=\left(
\begin{array}{ccc}
-1 & 0 & 0 \\
0 & 0 & 0 \\
0 & 0 & 0%
\end{array}\right)\left(
\begin{array}{c}
e_1 \\
e_2 \\
e_3
\end{array}%
\right)=\left(
\begin{array}{c}
-x_1+\sqrt{0.45}  \\
0  \\
0 %
\end{array}\right).
\end{equation}

\noindent Finally, $u$ receives the form

\begin{equation}\label{ccc2}
\left(
\begin{array}{c}
u_{1} \\
u_{2} \\
u_{3}%
\end{array}%
\right)=\left(
\begin{array}{c}
0 \\
-\widetilde{sgn}(x_{1})+x_{3}{sgn}(x_{1}) \\
-x_{1}^{2}+0.45
\end{array}%
\right) +\left(
\begin{array}{c}
-x_1+\sqrt{0.45} \\
0 \\
0%
\end{array}%
\right).
\end{equation}

\noindent After subtracting (\ref{sim_ne}) from (\ref{sim_master}), and replacing $u$, the obtained error system (\ref{liniar}) is
\begin{equation}\label{error_simi}
\begin{array}{l}
D_{\ast }^{q_{1}}e_{1}=e_{2}-e_1, \\
D_{\ast }^{q_2} e_{2}=-0.75e_2 , \\
D_{\ast }^{q_3} e_{3}=-0.45 e_3,%
\end{array}%
\end{equation}

\noindent with the error matrix
\begin{equation}\label{E}
E=\left(
\begin{array}{ccc}
-1 & 1 & 0 \\
0 & -0.75 & 0 \\
0 & 0 & -0.45%
\end{array}%
\right),
\end{equation}

\noindent which can also be calculated as $E=K+M$. The triangular matrix $E$ has eigenvalues $\Lambda=\{-1,-0.75,-0.45\}$ and $\alpha_{min}=\pi$ verifies the condition (\ref{stab3}) for asymptotical stability : $\alpha_{min}=\pi>1.492=q\pi/2$. 

The equilibrium $X_1^*$ is therefore stabilized, as we can observe from Figure \ref{fig7}a. In a similar way also the equilibrium $X_2^*$ can be stabilized (see Figure \ref{fig7}b).

\vspace{3mm}

\noindent \underline{\emph{Incommensurate case}} $q=(1,0.9,0.9)$

\noindent \emph{Equilibrium} $X_1^*$. The characteristic equation (\ref{car}) is
\begin{equation*}
\begin{split}
&P(\lambda):=det(diag [\lambda^{10},\lambda^{9},\lambda^{9}]-J_{X_1^*})\\
&=\lambda^{28}+6/5\lambda^{19}+27/80\lambda^{10}+3/5\sqrt{5}=0.
\end{split}
\end{equation*}

\noindent All 28 zeros of $P$ are plotted in Fig. \ref{fig9} a. $\alpha_{min}=0.147$ and, since $\alpha_{min}<0.157=\gamma\pi/2$ (with $\gamma=1/m=1/10$), $X_1^*$ is unstable (one can see that there are two (grey) roots outside the stability region $\Omega$ given by (\ref{staba}). In this case, the system behaves chaotically (Fig. \ref{fig6} b). If we chose

\begin{equation*}
M=\left(
\begin{array}{ccc}
-1&0&0 \\
-1&0&0 \\
0&0&0
\end{array}%
\right),
\end{equation*}

\noindent the controller becomes
\begin{equation}
\left(
\begin{array}{c}
u_{1} \\
u_{2} \\
u_{3}%
\end{array}%
\right)=\left(
\begin{array}{c}
0 \\
-\widetilde{sgn}(x_{1})+x_{3}{sgn}(x_{1}) \\
-x_{1}^{2}+0.45
\end{array}%
\right) +\left(
\begin{array}{c}
-x_1+\sqrt{0.45} \\
-x_1+\sqrt{0.45} \\
0%
\end{array}%
\right),
\end{equation}

\noindent and the error matrix $E$ in this case is
\begin{equation*}
\left(
\begin{array}{ccc}
-1&1&0 \\
-1&-0.75&0 \\
0&0&-0.45
\end{array}
\right),
\end{equation*}

\noindent To verify if $X_1^*$ has been stabilized, we have to find the roots of the characteristic equation (\ref{car})
\begin{equation*}
\lambda^{28}+6/5\lambda^{19}+\lambda^{18}+27/80\lambda^{10}+11/5\lambda^9+63/80=0.
\end{equation*}

\noindent The roots are plotted in Fig. \ref{fig9} b and $\alpha_{min}=0.241>0.157=\gamma\pi/2=\pi/20$. Now, all the roots are inside the stability region and therefore, $X_1^*$ is stabilized (Fig. \ref{fig10} a).

\noindent $X_2^*$ has been stabilized in the same way (Fig. \ref{fig10} b).

\vspace{3mm}

\noindent \underline{\emph{Nonhyperbolic equilibrium} $X_{3}^*(0,0,0)$}

\noindent In this case
\[
\tilde{J}_{\tilde{X}^*_3}=\left(
\begin{array}{ccc}
0 & 1 & 0 \\
0 & -0.75 & 0 \\
0 & 0 & -0.45
\end{array}%
\right),
\]

\noindent and $\Lambda=\{0,-0.750,-0.145\}$. Because of the zero eigenvalue, $X_3^*$ is a nonhyperbolic equilibrium and, as it is well known, to determine if a nonhyperbolic point is asymptotically stable or unstable, is a delicate question and Theorem \ref{stabt} does not apply (for example, we cannot determine the argument of the eigenvalue $\lambda=0$). In the neighborhood of a non-hyperbolic equilibrium point, it is not generally possible to find a homeomorphism that transforms the nonlinear flow to that of the linearization, and in this case the best way to have an answer to this question, is to use the Lyapunov method.

In this paper, the instability of $X_3^*$ has been deduced by numerical simulations.

\noindent Despite the fact that nonhyperbolicity causes some trouble in chaos control (such as for OGY-types controls for chaos which might fail \cite{DeBin2002}), the controller (\ref{ctrl_simi}), which in this case ($x_i^*=0$, $i=1,2,3$) has the form
\begin{equation}
\left(
\begin{array}{c}
u_{1} \\
u_{2} \\
u_{3}%
\end{array}%
\right)=\left(
\begin{array}{c}
0 \\
-\widetilde{sgn}(x_{1})+x_{3}{sgn}(x_{1}) \\
-x_{1}^{2}
\end{array}%
\right) +\left(
\begin{array}{c}
-x_1\\
0 \\
0%
\end{array}%
\right).
\end{equation}

\noindent stabilizes $X_3^*$, both for the commensurate case incommensurate case (Fig. \ref{fig11} a and \ref{fig11} b respectively).

\section{Conclusion}
In this paper we presented an algorithm to stabilize chaotic motions of a class of PWC systems of fractional order, by stabilizing the unstable equilibria. For this purpose, we transformed the discontinuous IVP into a continuous one to which the standard control algorithms apply. The approximation (sigmoid function) approximates globally the PWC components which appear in the system's mathematical model.

For the chaos control, we adopted one of simplest active control scheme, which stabilized the unstable equilibria.

For the numerical integration, the Adamas-Bashforth-Moulton scheme for fractional order DEs has been used.


\newpage

\begin{figure}
\begin{center}
  \includegraphics[clip,width=0.5\textwidth] {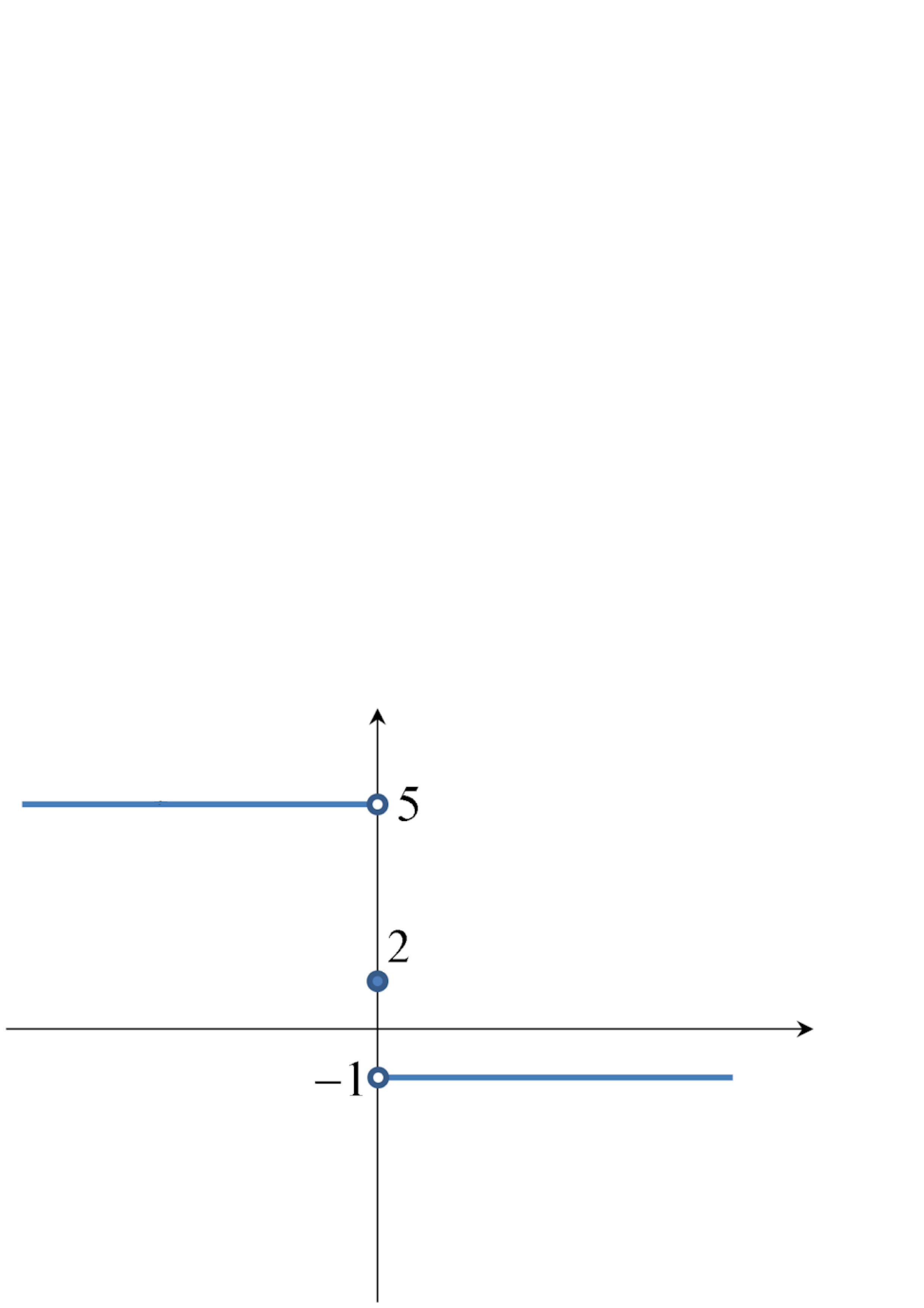}
\caption{Graph of PWC function $f(x)=2-3sgn(x)$.}
\label{fig1}
\end{center}
\end{figure}

\begin{figure}
\begin{center}
	\begin{tabular}{cc}
		\includegraphics[width=0.45\textwidth]{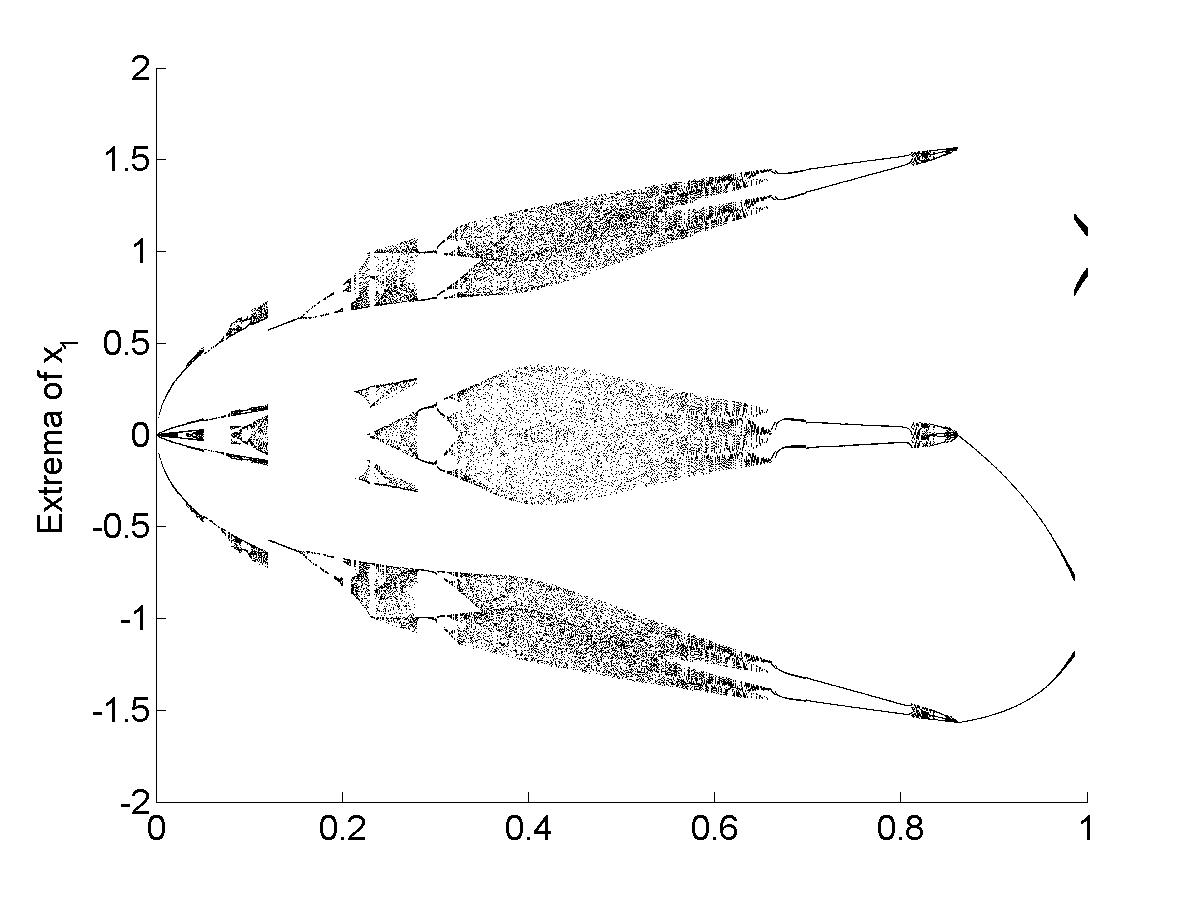} &
		\includegraphics[width=0.45\textwidth]{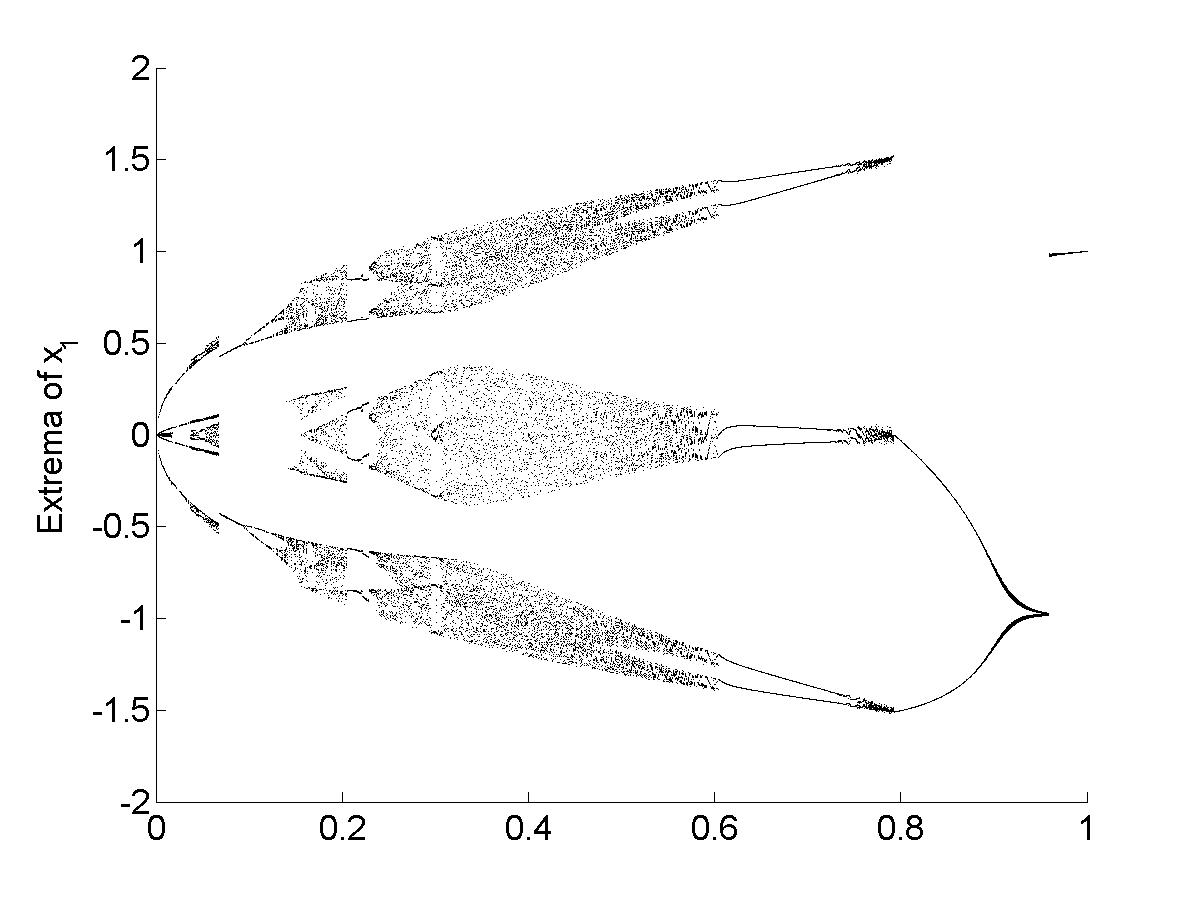} \\	\\	
		\includegraphics[width=0.45\textwidth]{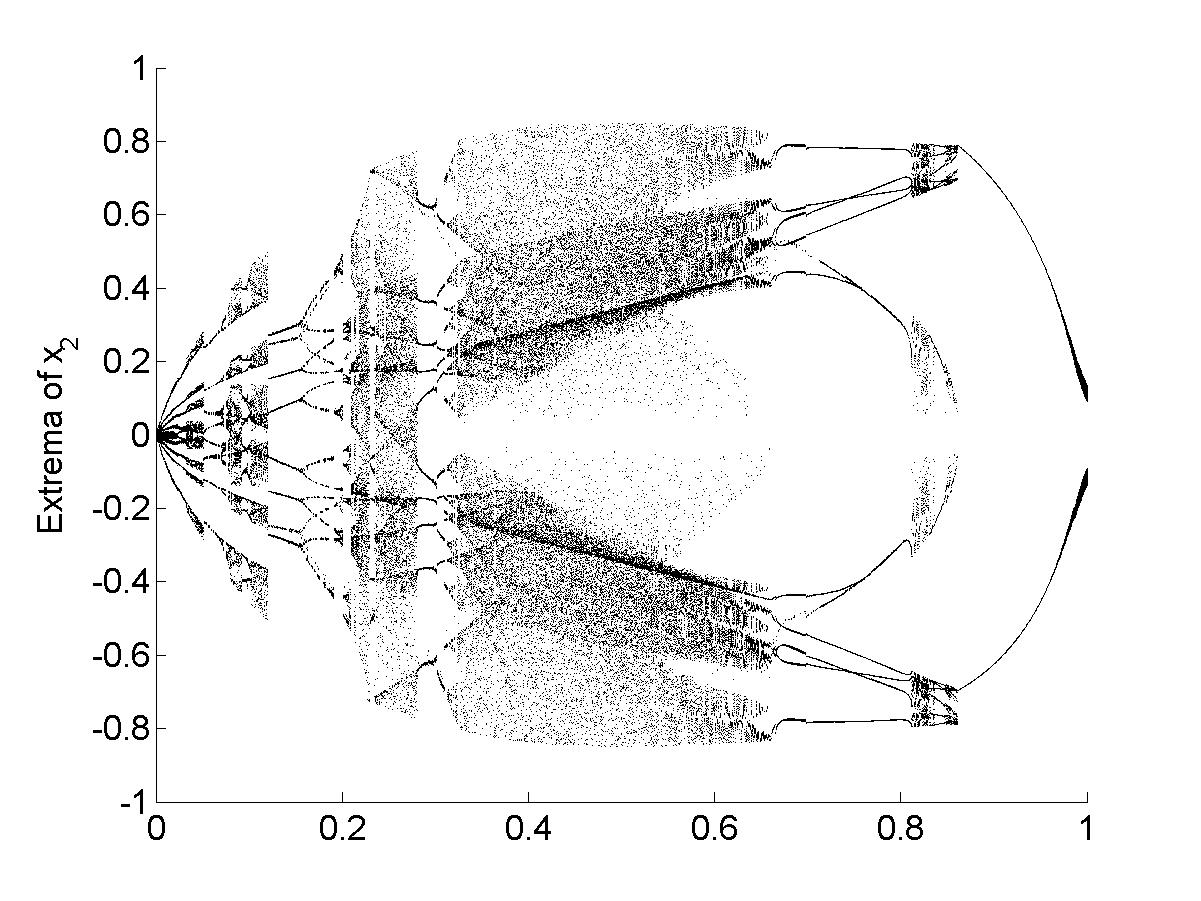} &
		\includegraphics[width=0.45\textwidth]{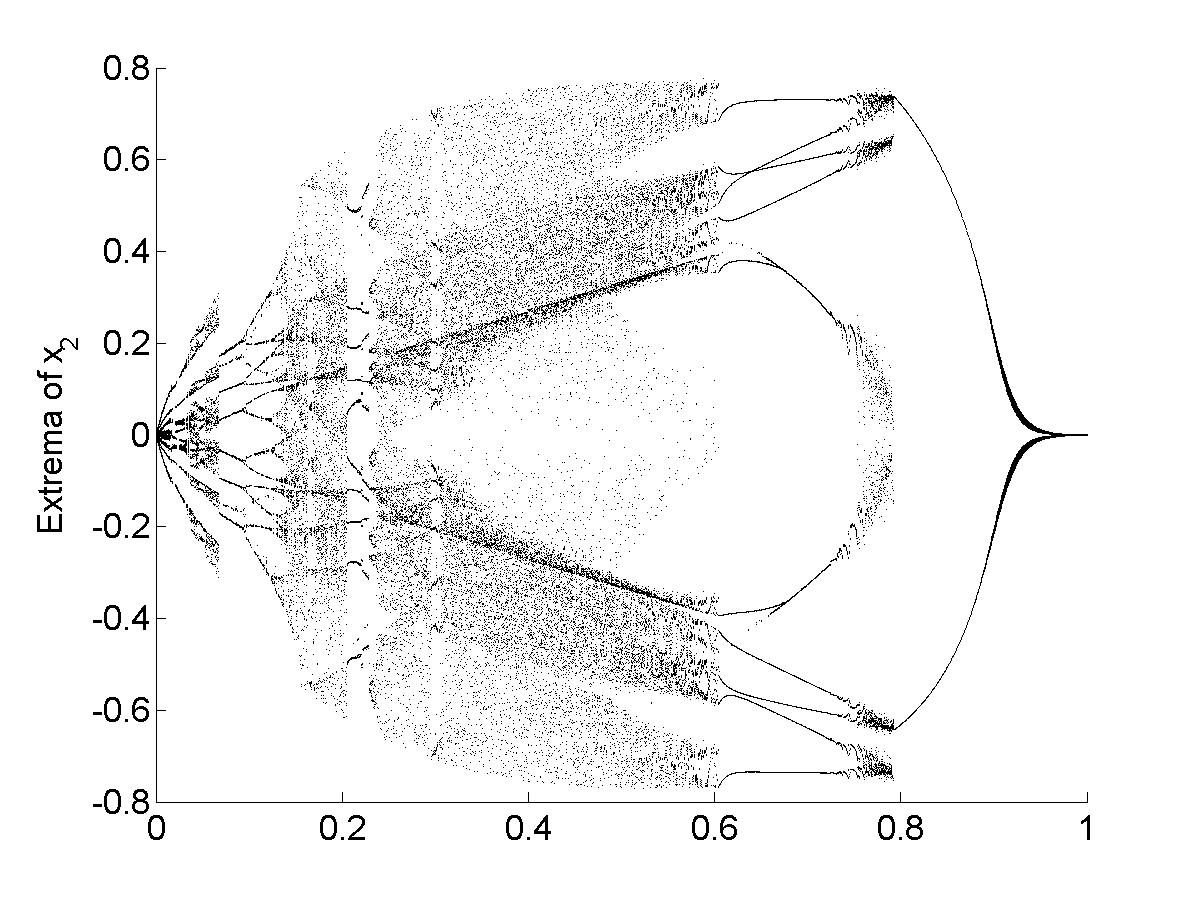} \\	\\	
		\includegraphics[width=0.45\textwidth]{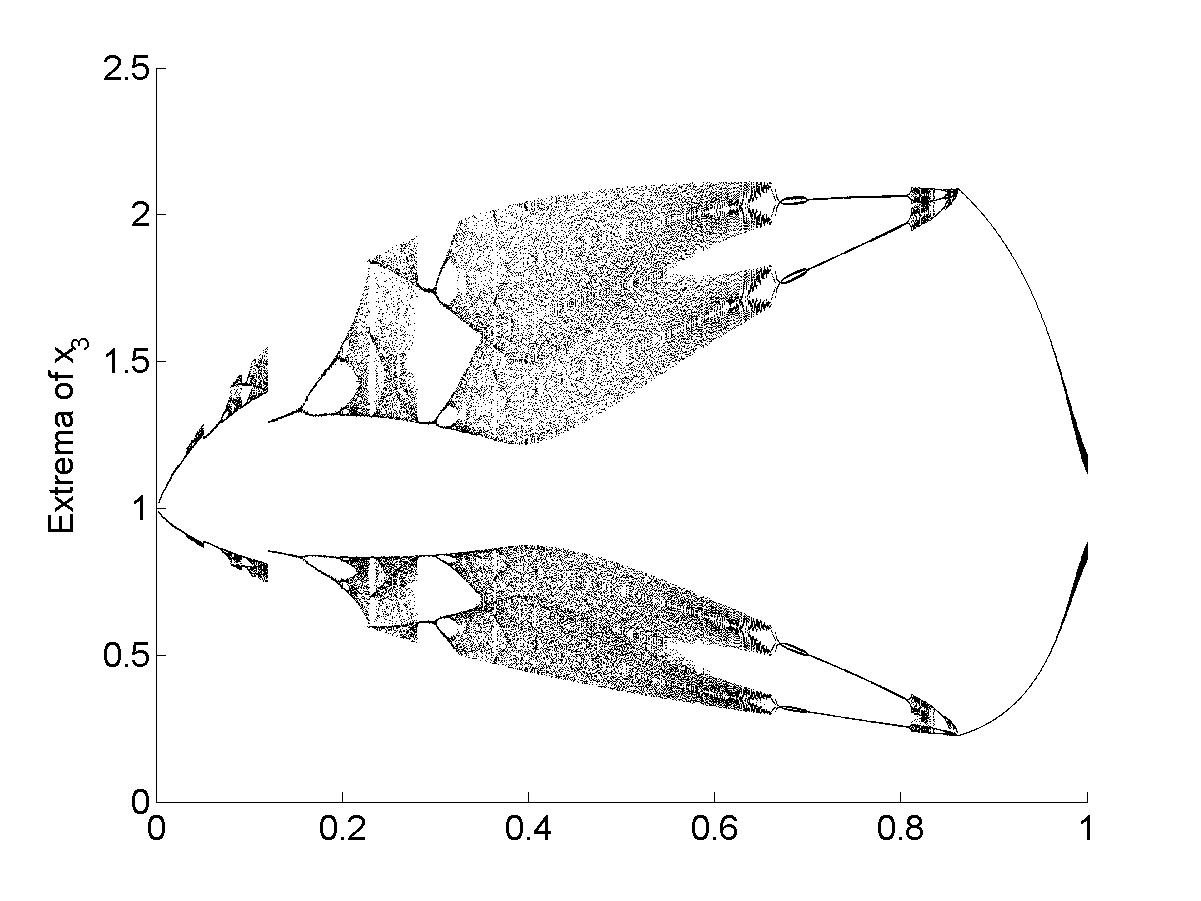} &
		\includegraphics[width=0.45\textwidth]{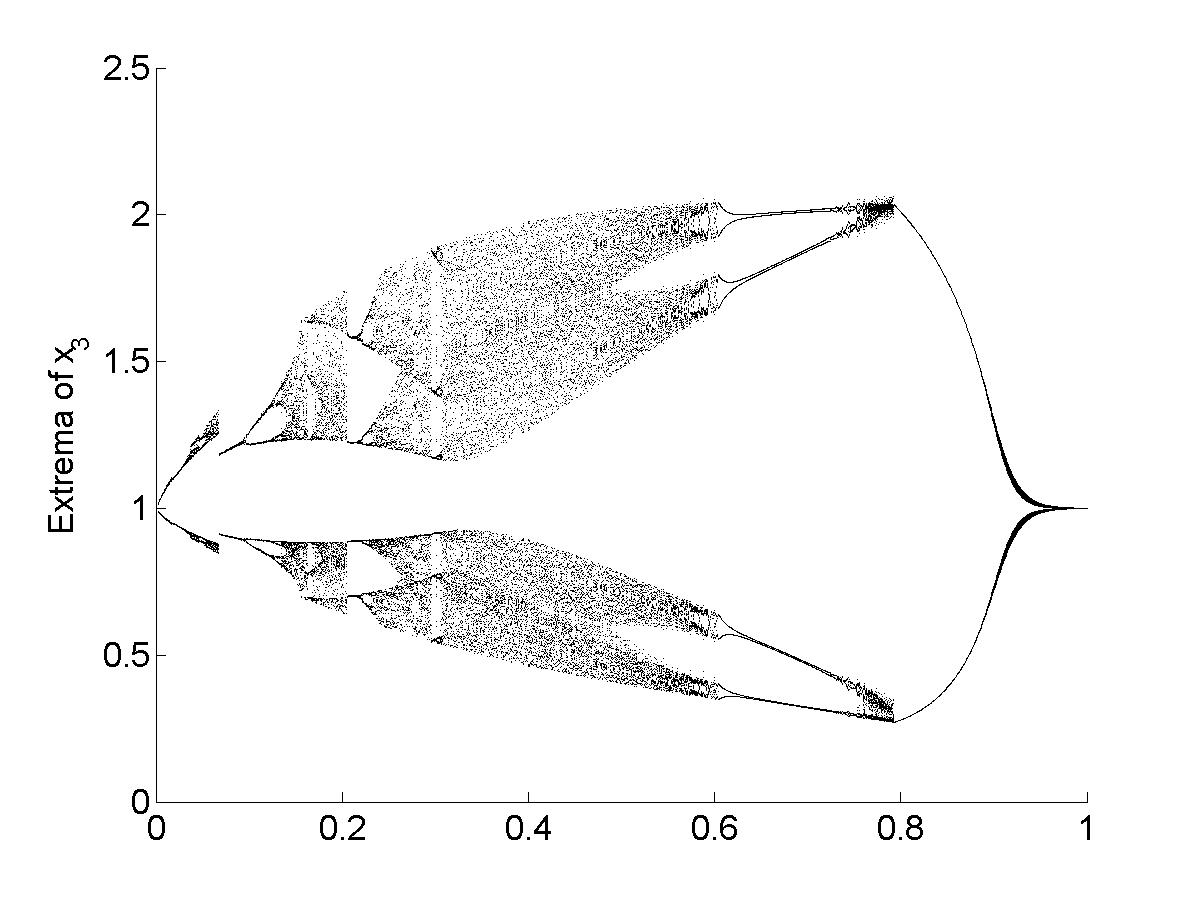} \\		
	\end{tabular}
\caption{Bifurcation diagrams (as the parameter $b$ varies in $[0,1]$) of extrema of the state variables for the Shimizu--Morioka's system. Left column: commensurate case $q_1=q_2=q_3=0.95$. Right column: incommensurate case $q_1=1$, $q_2=q_3=0.9$.}
\label{bd}
\end{center}
\end{figure}

\begin{figure}
\begin{center}
  \includegraphics[clip,width=0.5\textwidth] {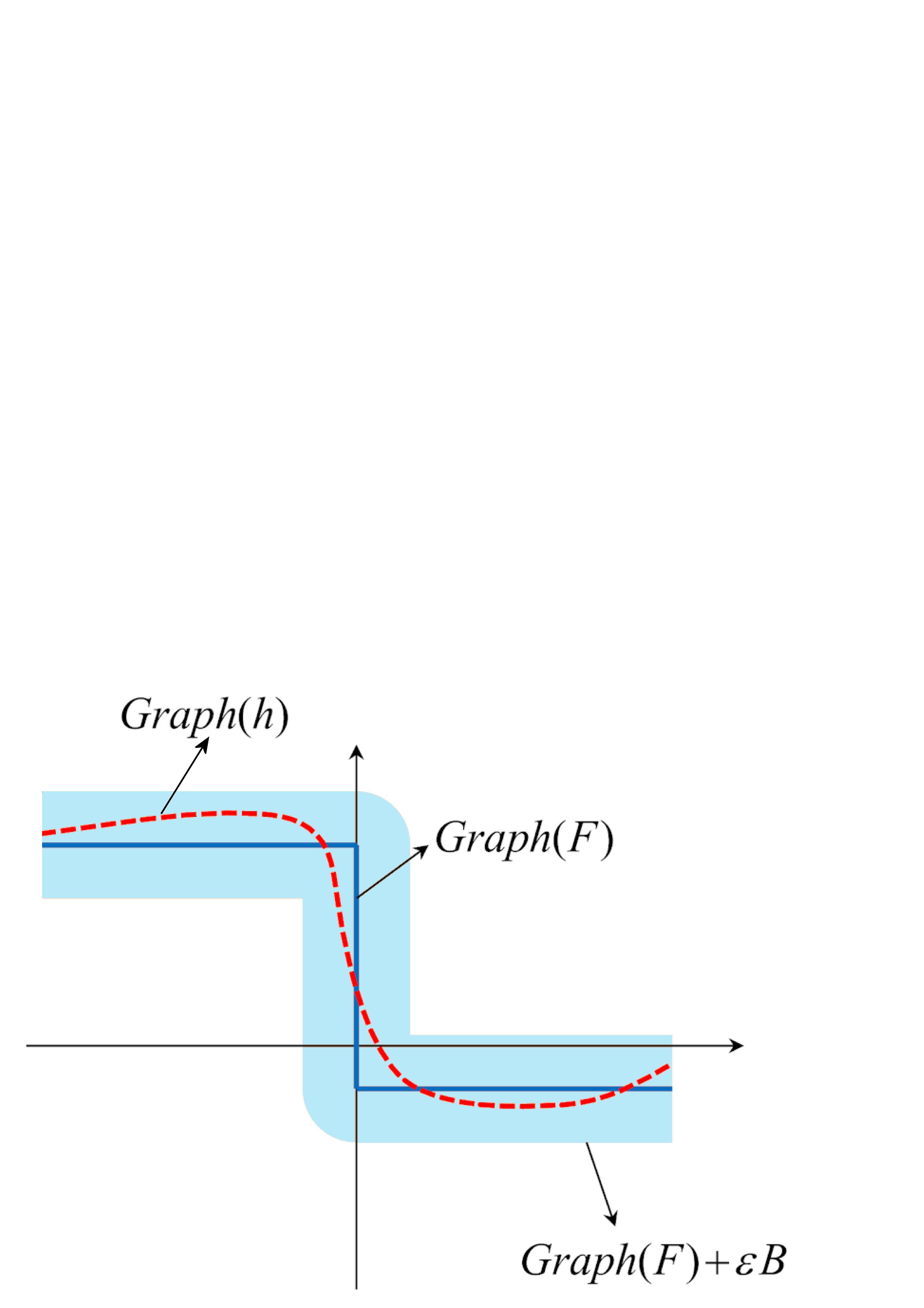}
\caption{Graph of a set-valued function (continuous line), his $\varepsilon$-neighborhood and a continuous approximation (dotted line).}
\label{fig2}
\end{center}
\end{figure}

\begin{figure}
\begin{center}
  \includegraphics[clip,width=0.9\textwidth] {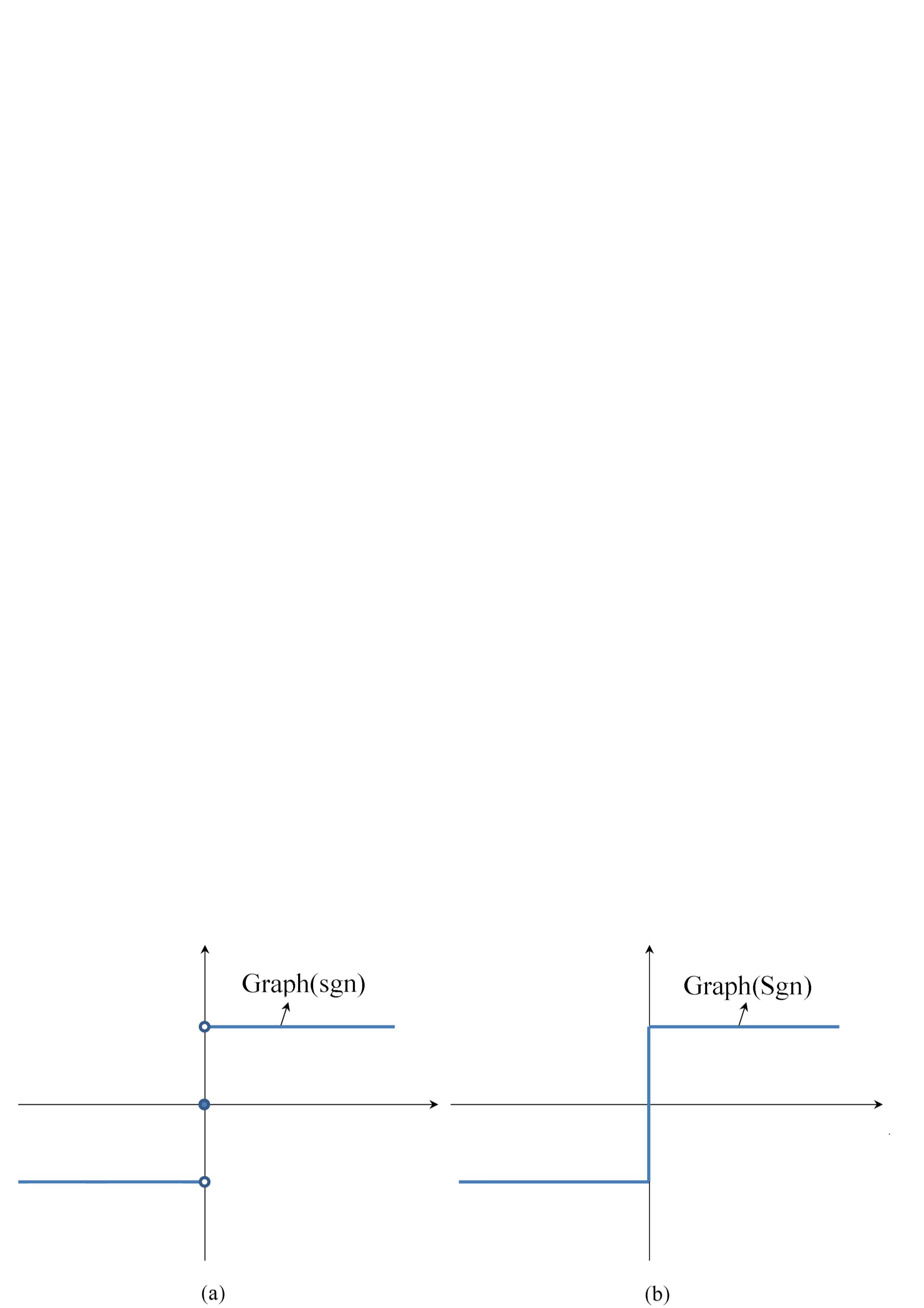}
\caption{a) Graph of $sgn$ function. b) Graph of the set-valued function $Sgn$.}
\label{fig3}
\end{center}
\end{figure}

\begin{figure}
\begin{center}
  \includegraphics[clip,width=0.6\textwidth] {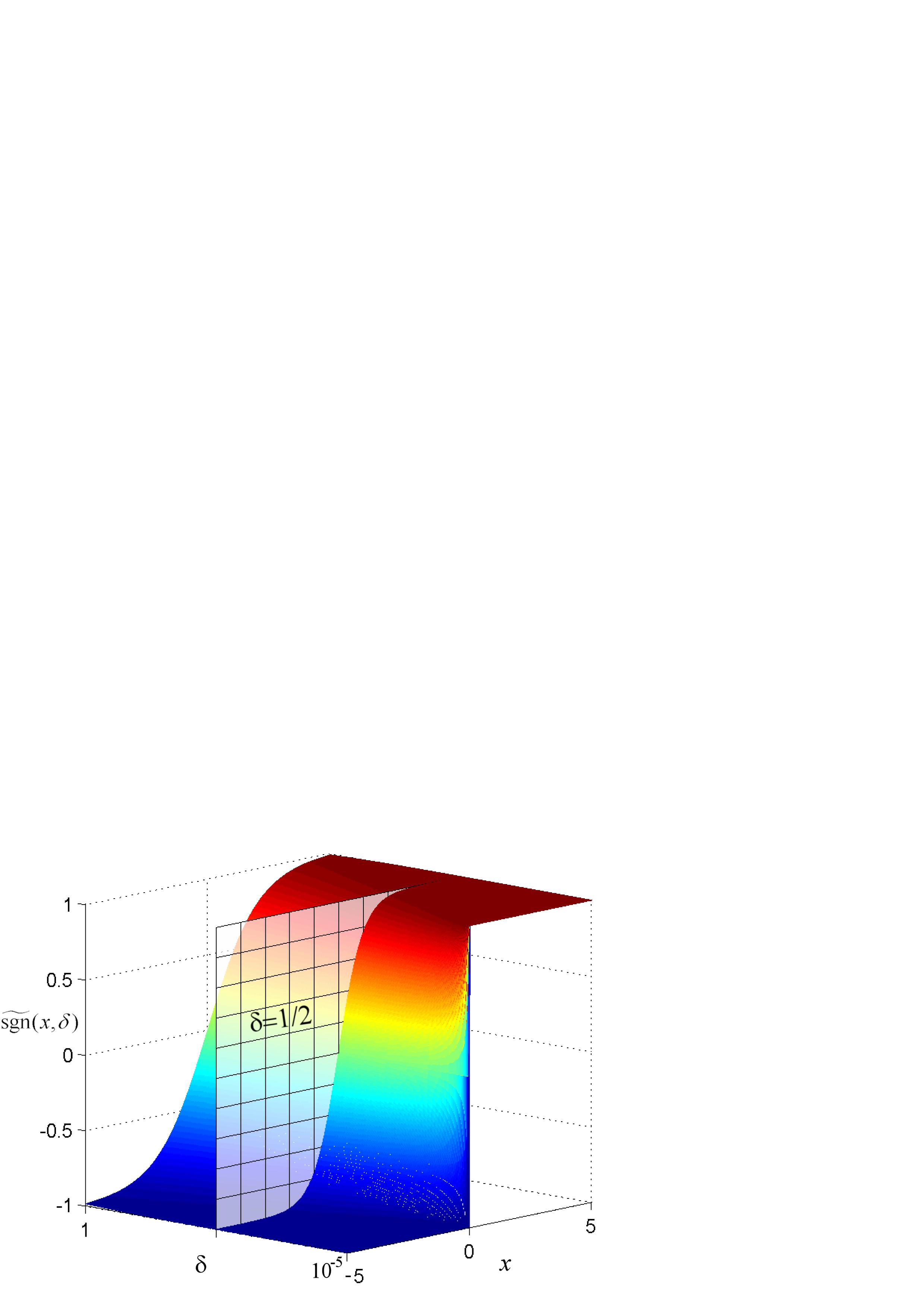}
\caption{Surface representing the family of functions $\widetilde{sgn}$ (\ref{h_simplu}), depending on $\delta$. The transversal plane reveals one $\widetilde{sgn}$ function, corresponding to $\delta=1/2$. }
\label{fig4}
\end{center}
\end{figure}

\begin{figure}
\begin{center}
  \includegraphics[clip,width=0.7\textwidth] {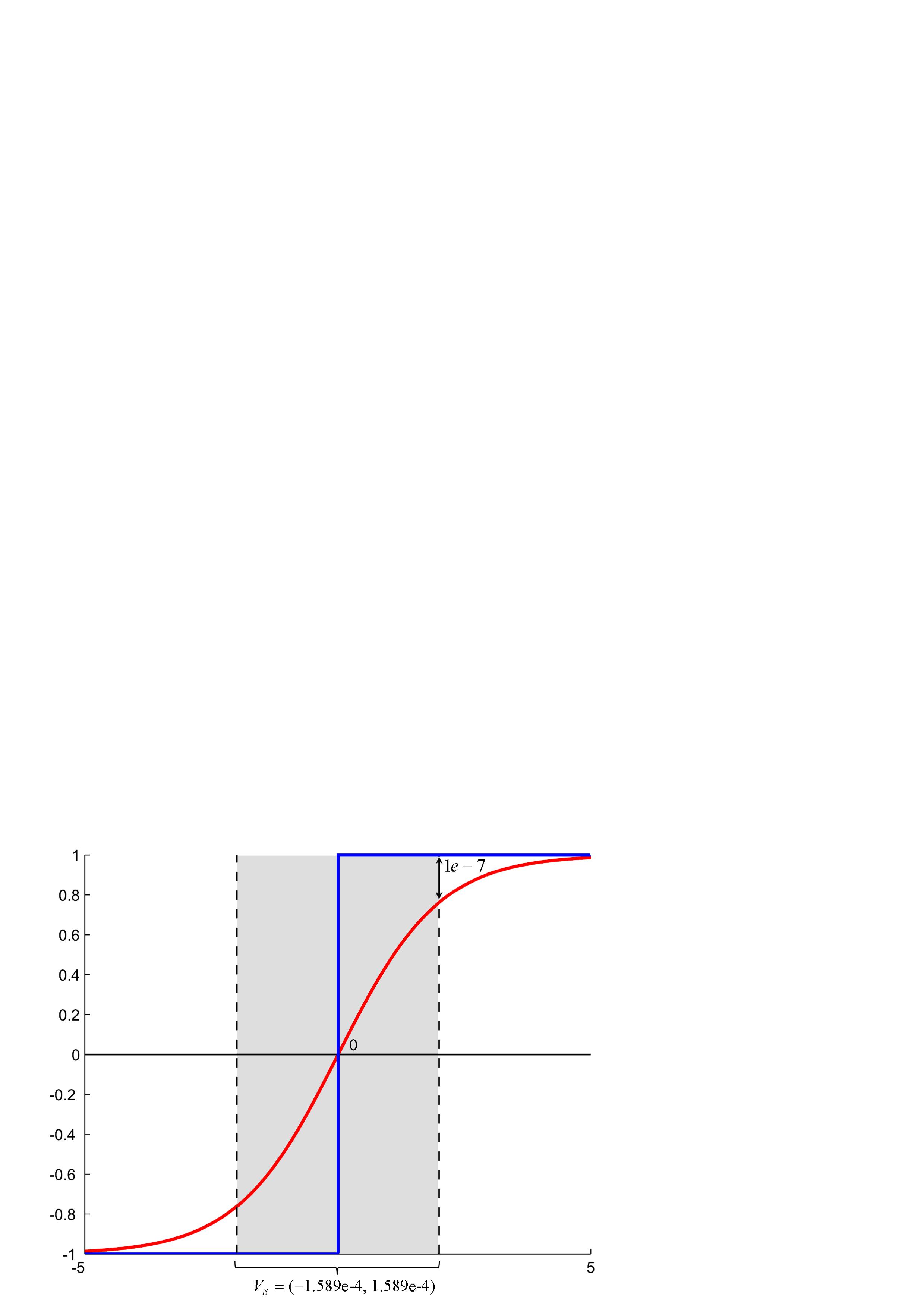}
\caption{Distance between $\widetilde{sgn}$ and $Sgn$ at $x=1.589e-4$. For clarity, $\mathcal{V}_{\delta}$ is drawn larger (see Remark \ref{cater}).}
\label{fig5}
\end{center}
\end{figure}

\begin{figure}
\begin{center}
  \includegraphics[clip,width=1\textwidth] {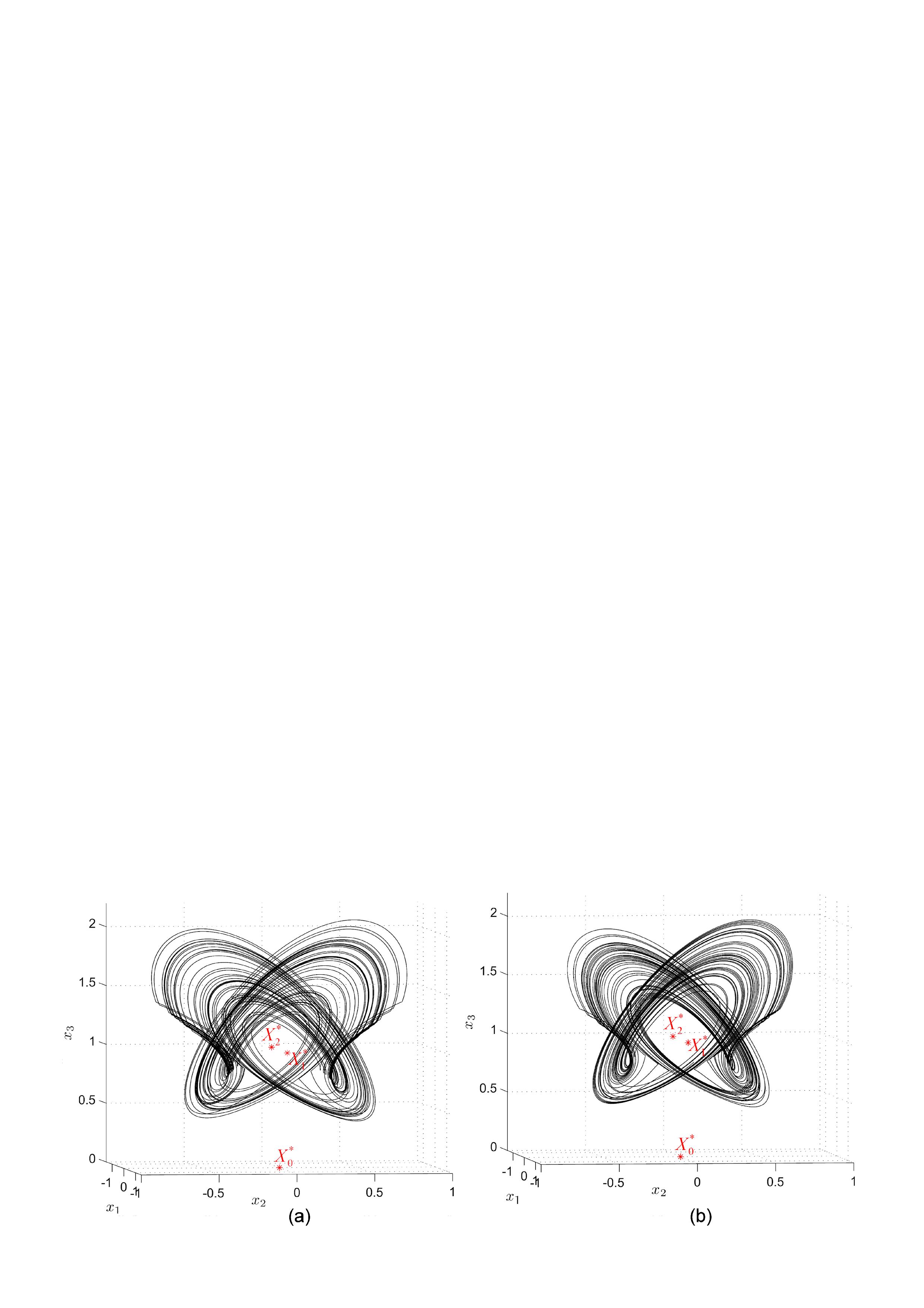}
\caption{Chaotic attractors of Shimizu--Morioka's system. a) Commensurate case $q_1=q_2=q_3=0.95$. b) Incommensurate case $q_1=1$, $q_2=q_3=0.9$.}
\label{fig6}
\end{center}
\end{figure}

\begin{figure}
\begin{center}
  \includegraphics[clip,width=0.5\textwidth] {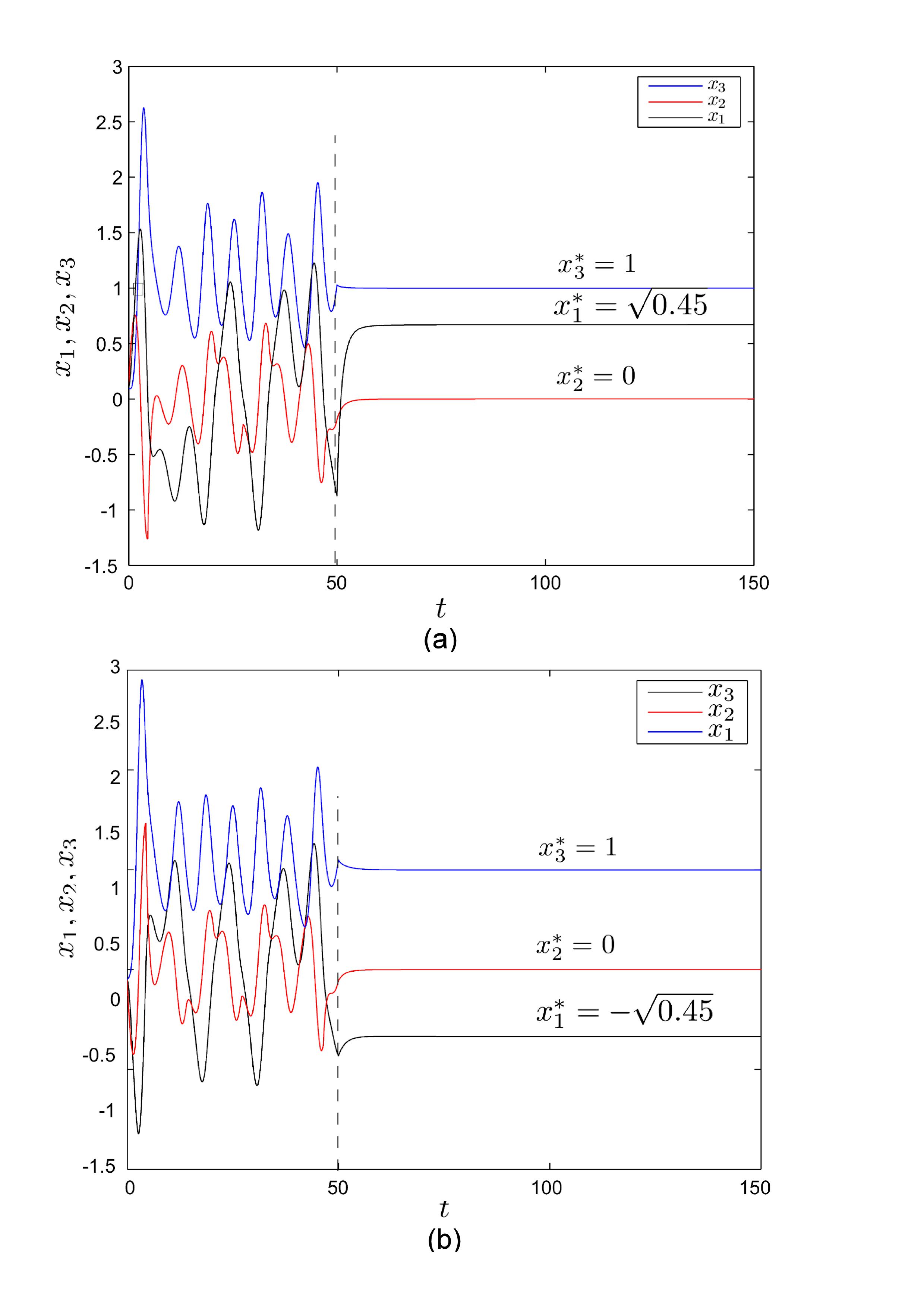}
\caption{Time series revealing the stabilization of the equilibrium points $X_{1,2}^*(\pm\sqrt{0.45},0,1)$ in the commensurate case $q=(0.95,0.95,0.95)$.  The control is activated at $t=50$. a) Equilibrium point $X_1^*(\sqrt{0.45},0,1)$. b) Equilibrium point $X_2^*(-\sqrt{0.45},0,1)$.}
\label{fig7}
\end{center}
\end{figure}

\begin{figure}
\begin{center}
  \includegraphics[clip,width=1\textwidth] {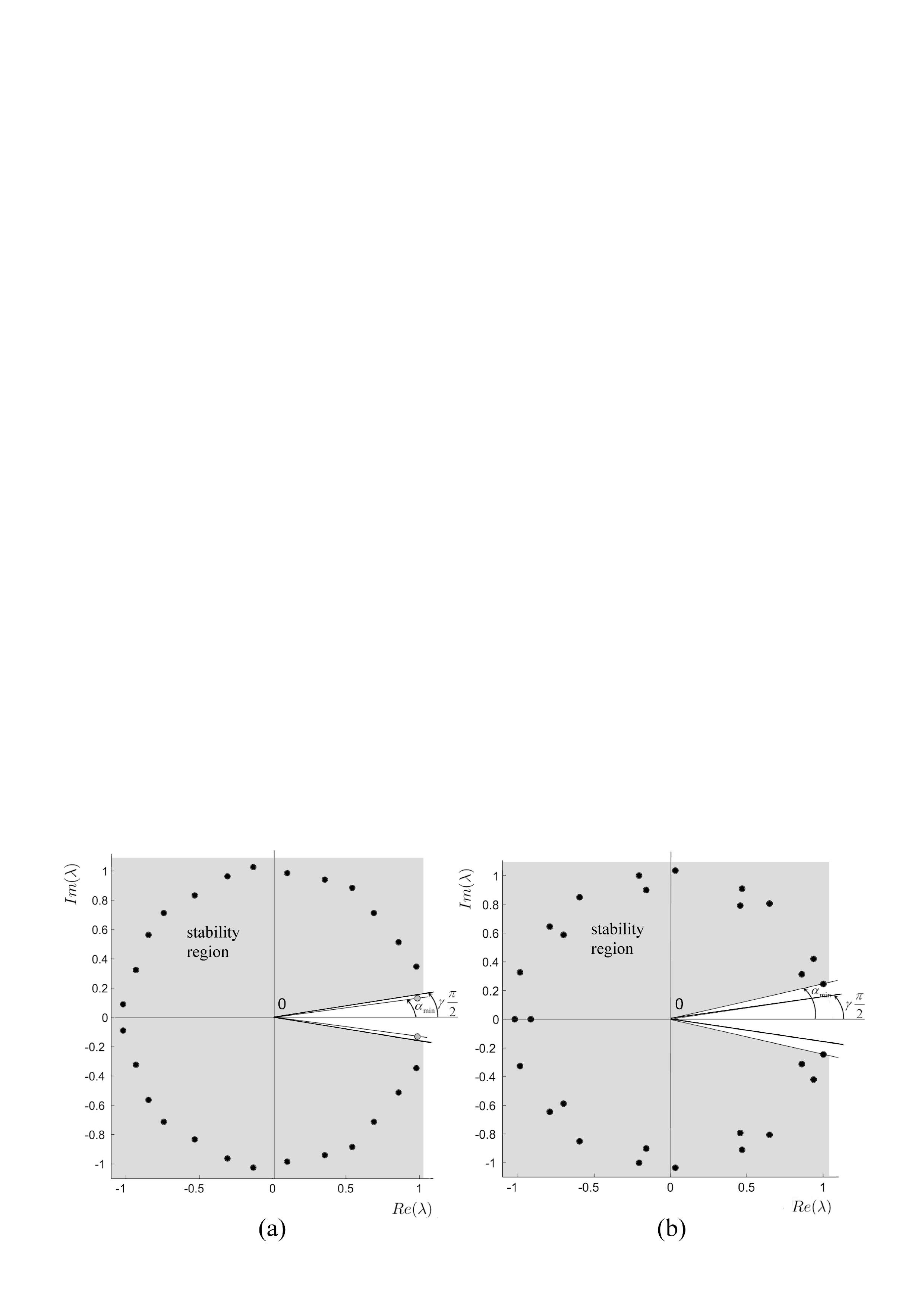}
\caption{Roots of the characteristic equation of $X_1^*$, for the incommensurate case $q=(1,0.9,0.9)$. a) Before control. b) After control.}
\label{fig9}
\end{center}
\end{figure}

\begin{figure}
\begin{center}
  \includegraphics[clip,width=0.5\textwidth] {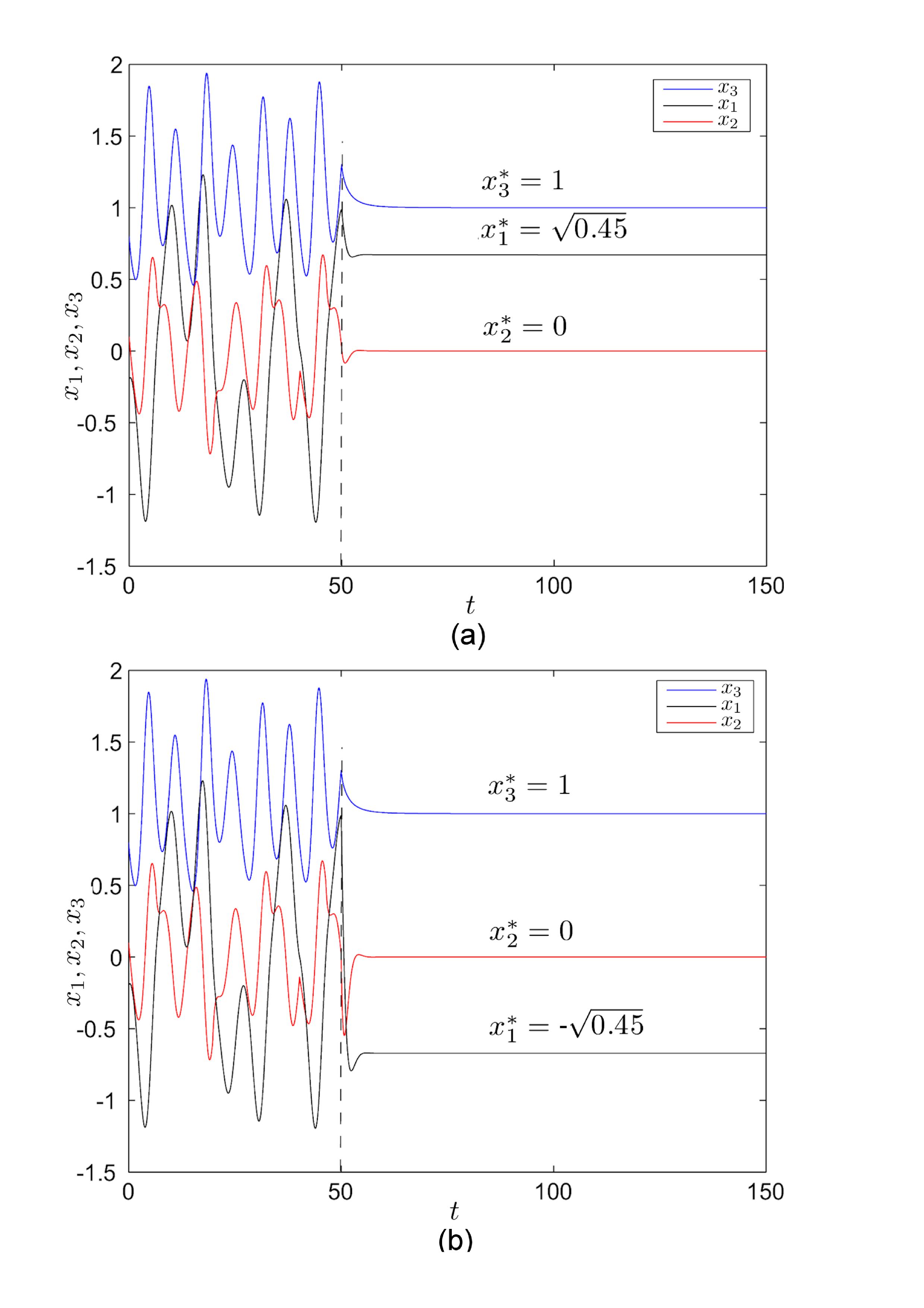}
\caption{Time series revealing the stabilization of the equilibrium points $X_{1,2}^*(\pm\sqrt{0.45},0,1)$ in the incommensurate case $q=(1,0.9,0.9)$.  The control is activated at $t=50$. a) Equilibrium point $X_1^*(\sqrt{0.45},0,1)$. b) Equilibrium point $X_2^*(-\sqrt{0.45},0,1)$.}
\label{fig10}
\end{center}
\end{figure}

\begin{figure}
\begin{center}
  \includegraphics[clip,width=0.5\textwidth] {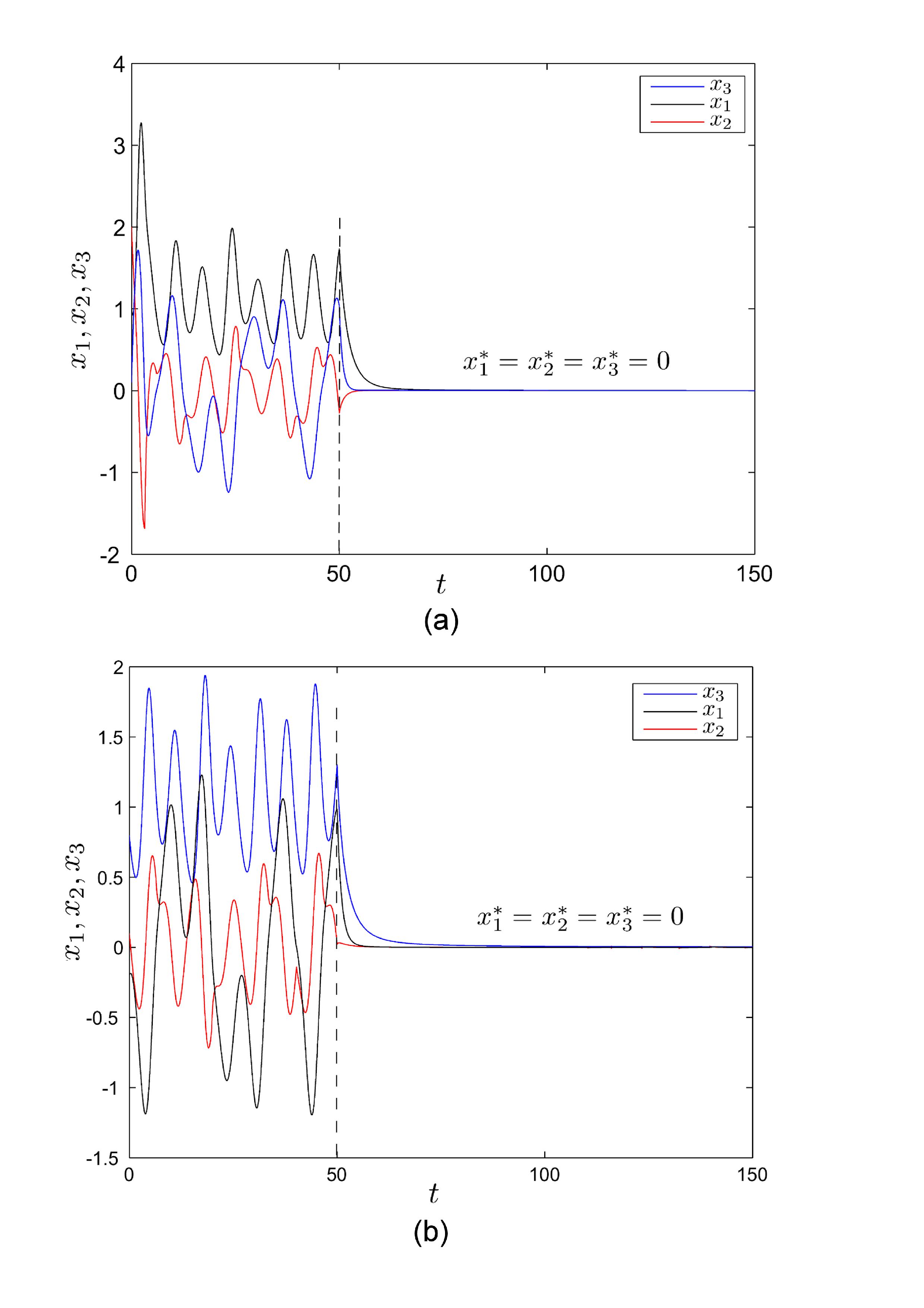}
\caption{Time series revealing the stabilization of the equilibrium points $X_{3}^*(0,0,0)$. The control is activated at $t=50$. a) Commensurate case $q=(0.95,0.95,0.95)$. b) Incommensurate case $q=(1,0.9,0.9)$.}
\label{fig11}
\end{center}
\end{figure}

\end{document}